\documentclass[3p,times,preprint]{elsarticle}
\usepackage[T1]{fontenc}
\usepackage[utf8]{inputenc}
\usepackage{graphicx}
\usepackage{hyperref}
\usepackage{amssymb}
\usepackage{amsmath,bm}
\usepackage{lmodern} 
\usepackage{textcomp} 
\usepackage{sectsty}
\usepackage{fancybox}
\usepackage{listings}
\usepackage{charter}
\usepackage[titletoc]{appendix}
\usepackage{epsfig}
\usepackage{tabularx}
\usepackage{enumitem}
\usepackage[ruled,vlined]{algorithm2e}
\usepackage{latexsym}
\usepackage{alltt}
\usepackage{setspace}
\usepackage{datetime}

\usepackage{ntheorem}

\newdateformat{mydate}{\monthname[\THEMONTH], \THEYEAR}
\def\qed{\hfill$\Box$}

\newproof{proof}{Proof}
\newtheorem{theorem}{\textbf{Theorem}}
\newtheorem{corollary}[theorem]{\textbf{Corollary}}

\newtheorem{proposition}[theorem]{\textbf{Proposition}}
\newtheorem{definition}[theorem]{\textbf{Definition}}

\newproof{pot1}{Proof of Theorem \ref{thm:constup}}

\newenvironment{statement}[1][Statement]{\noindent {\it #1}.}{}

\begin{document}

\begin{frontmatter}
	
	
	
	\title{Partially perfect hash functions for intersecting families}
	
	

	\author[label1]{Tapas Kumar Mishra \fnref{ack}}
	\ead{mishrat@nitrkl.ac.in}
	
	
	\cortext[cor1]{Corresponding author}
	\address[label1]{Department of Computer Science and Engineering, National Institute of Technology, Rourkela, 768009, India}

\begin{abstract}
%
Consider a large network with unknown number of nodes. Some of these nodes coordinate to perform tasks.
The number of such coordination groups is also unknown. The only information about the network available is that any two coordinating groups share at least $t$ nodes. To complete a particular task in a day,  at least $p$ nodes of the corresponding coordinating group must get different time slots out of the $r$ available slots per day. Is there a way of estimating the number of days required such that every coordinating group gets at least one day where it can complete the task? As it turns out, this problem is a special case of \textit{partially} perfect hash functions for intersecting families.
\end{abstract}

\end{frontmatter}

\section{Introduction}

A hash function $h$ from $[n]$ into $[b]$ is said to be a \textit{perfect}
with respect to a subset $S \subseteq [n]$ provided that $h$ is one-to-one on $S$. A
collection $H$ of functions from $[n]$ into $[b]$ is called a $(b, k)$-family of perfect
hash functions provided that for each subset $S\subseteq [n]$ of size $k$, there is a function $h\in H$
which is perfect with respect to $S$. A $(b, k)$-family of perfect hash functions provides
a means for storing subsets of size $k$ into tables with $b$ cells. Fredman and Komlos \cite{fredman1984size} define $Y(b, k, n)$ to be the minimum
size of any $(b, k)$-system. They proved that 
\begin{align}
bg(\alpha) + \mathcal{O}(\log b) \leq \log Y(b, k, n)
 \leq bg (\alpha) + \log \log n + \mathcal{O} (\log b) \nonumber \\
\text{where } g(a)= (1-\alpha) \log (1-\alpha)+\alpha \log e.
\end{align}


Consider the following notion of partially perfect hash functions. Let $h: [n] \rightarrow [b]$ be a hash function. 
For a subset $S \subseteq [n]$, let $h(S):= \{h(s)|s \in S\}$. 
\begin{definition}
A hash function $h$ from $[n]$ into $[b]$ is said to be a \textit{partially $p$-perfect}
with respect to a subset $S \subseteq [n]$ provided that $|h(S)|\geq \min(p,|S|)$.
A collection $H$ of functions from $[n]$ into $[b]$ is called a $(b,k)$-family of partially $p$-perfect
hash functions provided that for each subset $S\subseteq [n]$ of size $k$, there is a function $h\in H$
which is partially $p$-perfect with respect to $S$.
\end{definition}

Let $H_1, H_2, \ldots , H_t$ be $p$-partite $r$-uniform hypergraphs such that $H_1 \cup H_2 \cup \ldots \cup H_t = K_n^r$, where $K_n^r$ denotes a complete $r$-uniform hypergraph on $n$ vertices. Then, we have the following lower bound on $t$ (see \text{\cite{Radhakrishnan01entropyand}} for details).
\begin{align}
t \geq \frac{\binom{n}{r-2}(n-r+2)\log(n-r+2)}{(k-r+2)(n/k)^{r-1}\binom{k}{r-2}\log(k-r+2)}.
\label{1:eq}
\end{align}
See Radhakrishnan \cite{Radhakrishnan01entropyand} for related problems and results.
The connection between $t$ and a minimum cardinality $(b,k)$-family of partially $p$-perfect hash function is immediate.

We study an interesting variation of partially $p$-perfect hash functions for \textit{intersecting families}.
A hypergraph is called $t$-intersecting if every pair of hyperedges share at least $t$ vertices. The set of hyperedges of a $t$ intersecting hypergraph is called as a $t$-intersecting family.
\begin{definition}
A $(G,p,b)$ system for a hypergraph $G$ is a collection $H$ of partially $p$-perfect
hash functions such that for each hyperedge $e \in G$, there is a function $h\in H$
which is partially $p$-perfect with respect to $e$. 
\end{definition}
Let $\lambda(G,p,b)$ denote the minimum cardinality of a $(G,p,b)$ system where $G$ is a hypergraph.
Let $\lambda(t,p,b)$ denote the maximum of $\lambda(G,p,b)$ where $G$ is a $t$-intersecting family. In this paper, we study the parameter $\lambda(t,p,b)$ in detail.

\subsection{Main result}

In Section \ref{sec:2.1}, we prove the following result.
\begin{theorem}\label{thm:1}
	Let $t \geq p \geq 2$ and $b \geq p$ such that $b > (p-1)(t+1)$.
	Let  $x$ denote an integer such that $\binom{b}{p-1}(\frac{p-1}{b})^{tx}< 1$.
	Then, $\lambda(t,p,b) \leq x$. In particular, if $b > (p-1)(t+1)$,  for every $t$-intersecting hypergraph $G$, there exists a pair of hash functions from $[n]$ to $[b]$  that constitute a $(G,p,b)$ system.
\end{theorem} 
The theorem is interesting in the sense that provided $b > (p-1)(t+1)$, every $t$-intersecting hypergraph $G$ has a $(G,p,b)$ system of size at most 2. However, when $b \leq (p-1)(t+1)$, the study of the behaviour of $\lambda(t,p,b)$ remains open.

\section{Bounds on $\lambda(t,p,b)$ }

The fact that $\lambda(t,p,b)\geq \lambda(t+1,p,b)$ follows from the fact that any $t+1$-intersecting family is also $t$-intersecting.
Moreover, $\lambda(t,p,b)$ is non-increasing with increasing $b$ as well, i.e. $\lambda(t,p,b)\geq \lambda(t,p,b+1)$ due to the same containment argument. Further, $\lambda(0,p,b) = \infty$ - if no restriction on the minimum size of intersection between hyperedges is there, then $\lambda(0,p,b)$ must depend on the ground set size  as evident from Equation \ref{1:eq}.

\begin{proposition}\label{prop:p1}
\[	\lambda(1,2,2) =2. \]
\end{proposition}	
	
\begin{proof}
To prove the lower bound, let $V=\{1,2,3\}$, $G=\{\{1,2\},\{2,3\},\{1,3\}\}$.
$G$ is 1-intersecting. Let $h_1: V \rightarrow \{c1,c2\}$ be a hash function. It is easy to see that at least two elements of $V$ get the same hash value and hence, one hyperedge in $G$ is not partially $2$-perfect.
For the upper bound, consider any $1$-intersecting hypergraph $G$ on the ground set $V$. Pick an hyperedge $e$ such that there is no hyperedge $e'$ with $e' \subset e$. Let $h_1: V \rightarrow \{c1,c2\}$ be a hash function such that each element of $e$ get the hash value $c1$, and 
each element not in $e$ get the hash value $c2$. From the restriction on $e$, it follows that for every $e' \neq e$,
$h_1$ is partially $2$-perfect. For hyperedge $e$, we can choose another hash function.
\qed

\end{proof}

\begin{proposition}\label{prop:p2}
\[	\lambda(t,2,2)=1 \text{ for $t \geq 2$}. \]
\end{proposition}
\begin{proof}
	In what follows we prove that $\lambda(2,2,2)=1$ and from non-increasing property of $\lambda(t,p,b)$ with increasing $t$, the proof follows.
	Consider the following greedy algorithm for obtaining the hash function $h_1$, given a ground set $V$ and a $t$-intersecting family $G$.
	Assign the vertices with hash value $c1$ unless making it $c1$ makes every vertex in some $e \in G$ receiving the same hash value - in this case assign the vertex with hash value $c2$. For the sake of contradiction, assume that every vertex in $e$ gets the same hash value. Due to our algorithm, all vertices of $e$  gets the same hash value $c2$. Let $x$ be the last vertex getting hash value $c2$ belonging to $e$. The algorithm assigns $x$ with  hash value $c2$ due to the fact that all other vertices of some other hyperedge $e'$ containing $x$ received $c1$. As the hypergraph is 2-intersecting, $e \cap e'=\{x,y\}$, for some other vertex $y$. From the above argument, $y$ must have received the hash value $c1$ which contradicts our assumption that every vertex in $e$ gets the same hash value. This completes the proof.
	\qed
\end{proof}

\begin{proposition}\label{prop:p3}
\[\lambda(1,2,b) =1\text{ for $b \geq 3$}.\]
\end{proposition}
\begin{proof}
To see that $\lambda(1,2,b)=1$, first observe that $\lambda(1,2,3) \geq \lambda(1,2,b) $ for $r \geq 3$. So, in what follows, we 
prove that $\lambda(1,2,3)=1$; this suffices to prove  $\lambda(1,2,r) =1$.
Given a ground set $V$ and a $t$-intersecting family $G$,
pick a hyperedge $e \in E$ such that there is no hyperedge $e' \in E$ with $e' \subset e$. 
Assign few vertices in $e$ hash value $c1$, other vertices in $e$ hash value $c2$, and assign every other vertex in $V \setminus e$ hash value $c3$. This constitutes the desired hash function.
\qed
\end{proof}

As discussed earlier, $\lambda(t,p,b) \geq \lambda(t+1,p,b)$. But what happens if we modify parameters $t$ and $p$ simultaneously? In fact the monotonicity property changes, which we state in the following theorem.
\begin{theorem}
	$\lambda(t+1,p+1,b) \geq \lambda(t,p,b) $.
\end{theorem}
\begin{proof}
	Consider a $t$-intersecting hypergraph $G$ with minimum cardinality of a partial $(G,p,b)$ system exactly $\lambda(t,p,b)$. 
	We make it $k$-uniform by adding extra vertices.
	Our first task is to get a $t+1$-intersecting hypergraph, which we can achieve by adding a vertex $v$ to the vertex set $V(G)$ and adding $v$ to every hyperedge $e \in E(G)$. Let the new hypergraph be $G'$. Note that $\lambda(G',p+1,b)$ is a lower  bound  for $\lambda(t+1,p+1,b)$. So, all we need to show is $\lambda(G',p+1,b) \geq \lambda(t,r,p)$ in order to complete the proof.
	Let $x=\lambda(G',p+1,b)$. 
	Let $H'$ be a $(G',p+1,b)$ system of size $x$. 
	We construct another family $H$ of partial $(G,p,b)$ system by keeping the same hash values as in $H'$ on $V' \setminus \{v\}$. 
	It follows that $\lambda(G,p,b) \leq x$.
	Therefore,  $\lambda(t,p,b) = \lambda(G,p,b) \leq \lambda(G',p+1,b) \leq \lambda(t+1,p+1,b)$.
	\qed
\end{proof}
This bound implies a rather strange property of $t$-intersecting hypergraphs which is given as a corollary below.
\begin{corollary}\label{cor:2}
	$\lambda(t,p,b) \geq \lambda(t-1,p-1,b) \geq \ldots \geq \lambda(0,p-t,b) = \infty$ for $p-t \geq 2$.
\end{corollary}

So we focus on the cases where $t \geq p-1$. We start our analysis with the case of $t=p-1$. Consider a complete $k$-uniform hypergraph $G(V,E)$, where $|V|=  k(1+x)$, for some $0< x \leq 1$ and some positive integer $i$. Assume $G$ is $p-1$ intersecting. 
Consider a hash function $h_1:V \rightarrow k$.
The most frequent $p-1$ hash values are assigned to at least $ (p-1)\lfloor \frac{k(1+x)}{k} \rfloor=2p-2$ vertices (let this set of vertices be $V_1$). So for any $k \leq 2p-2$, there exists a hyperedge that is not partially $p$ perfect with respect to $h_1$. 
As a result, at least one more hash function is needed in order to cover that hyperedge. 
All we need to show now is there exists a $p-1$ intersecting complete $k$-uniform hypergraph on 
$(1+x)k$ vertices, where $k\leq 2(p-1)$ and $0 < x \leq 1$. Choosing $k=2p-2$, $x=0.5$, we get $|V|=3p-3$. 
Observe that the number of vertices shared between any two hyperedges is at least $2*(2p-2)-3p-3=p-1$, 
hence the hypergraph is $p-1$ intersecting. Consequently, we have the following lower bound.

\begin{theorem}
	$\lambda(p-1,p,b) \geq 2 $ for $p\leq b \leq 2p-2$.
\end{theorem}

For the general case of $t\geq p$, consider a $pt$-uniform complete hypergraph $G$. 
In order to make $G$ $t$-intersecting, we make $|V| = 2pt-t$. 
Consider a hash function $h_1:V \rightarrow 2p-3$.
The most frequent $p-1$ hash values are assigned to at least $\lceil (p-1)\frac{2pt-t}{2p-3}\rceil=\lceil \frac{2p^2t-3pt+t}{2p-3}\rceil=\lceil pt+\frac{t}{2p-3}\rceil > pt$ vertices (note that if the hash function is $h_1:V \rightarrow 2p-2$, then we may not get this guarantee). Consequently, there exist at least one hyperedge that receives at most $p-1$ hash values by $h_1$, which needs at least one more coloring. 

\begin{theorem}
	$\lambda(t,p,b) \geq 2 $ for $p\leq b \leq 2p-3$, $t \geq p$.
\end{theorem}

\subsection{An upper bound}
\label{sec:2.1}
For a fixed $0 < p < $1, the $p$-biased measure of a family $\mathcal{F}$ over $[n]$ is $\mu_p(\mathcal{F}) := Pr_S[S \in F]$, where the
probability over $S$ is obtained by including each element $i \in [n]$ in $S$ independently with probability $p$.
Such a set $S$ is called a $p$-biased subset of $[n]$. The combined results of Dinur and Safra \cite{dinur2005hardness} and Friedgut \cite{friedgut2008measure}
gives the following theorem.

\begin{theorem}\label{thm:2}
Fix $t \geq 1$. Let $\mathcal{F}$ be a $t$-intersecting family. For any
$p < \frac{1}{t+1}$ , the $p$-biased measure of $\mathcal{F}$ is bounded by $\mu_p(\mathcal{F}) \leq p^t$.
\end{theorem}

The following is an easy corollary of the above theorem.

\begin{corollary}\label{cor:1}
Fix $t \geq 1$. Let $\mathcal{F}$ be a $t$-intersecting family. For any $p < \frac{1}{t+1}$, the probability that a $p$-biased
subset of $[n]$ contains a set $S \in \mathcal{F}$ is at most $p^t$.
\end{corollary}

\begin{proof}
Consider a $t$-intersecting family $\mathcal{F}$ and let $\mathcal{T}=\{T \subseteq [n]| F \in \mathcal{F} \text{ and } F \subseteq T\}$.
Note that $\mathcal{T}$ is $t$-intersecting. Consider a $p$-biased subset $S \subseteq [n]$ for some $p < \frac{1}{t+1}$.
Observe that $S$ contains a set of $\mathcal{F}$ if and only if $S \in \mathcal{T}$. From Theorem \ref{thm:2}, we have
$\mu_p(\mathcal{T}) \leq p^t$. So, the probability that a $p$-biased
subset of $[n]$ contains a set $S \in \mathcal{F}$ is at most $p^t$.
\qed
\end{proof}

We give an upper bound on $\lambda(t,p,b)$ using the above corollary.\\

\begin{statement}[\bf Statement of Theorem \ref{thm:1}]
Let $t \geq p \geq 2$ and $b \geq p$ such that $b > (p-1)(t+1)$.
Let  $x$ denote an integer such that $\binom{b}{p-1}(\frac{p-1}{b})^{tx}< 1$.
Then, $\lambda(t,p,b) \leq x$. In particular, if $b > (p-1)(t+1)$,  for every $t$-intersecting hypergraph $G$, there exists a pair of hash functions from $[n]$ to $[b]$  that constitute a $(G,p,b)$ system.
\end{statement} 

\begin{proof}
Consider a $t$-intersecting hypergraph $G([n],E)$.
Let  $h_1, h_2, \ldots, h_x$ denote $x$ independent hash functions, where each $h_i:[n] \rightarrow b$ is a random hash function in which each vertex is assigned a hash value independently and uniformly at random from $[b]$.
Let $H$ denote a set of $p-1$ hash values.
The set $S_1 \subseteq [n]$ receiving one of the $p-1$ hash values in $H$ by the hash function $h_1$ is a random subset of $[n]$ where each element is included independently with probability $p=\frac{p-1}{b} < \frac{1}{t+1}$.  
Using Corollary \ref{cor:1}, the probability that $S_1$ contains an hyperedge from $G$ is
at most $(\frac{p-1}{b})^t$.
Similarly, the probability that $S_2$, the subset receiving one of the $p-1$ hash values in $H$ by the hash function $h_2$, contains an hyperedge from $G$ is
at most $(\frac{p-1}{b})^t$.
Continuing in this fashion, the probability that $S_x$, the subset receiving one of the $p-1$ hash values in $H$ by the hash function $h_x$, contains an hyperedge from $G$ is at most $(\frac{p-1}{b})^t$.
So, the probability that each vertex of any hyperedge $E$ from $G$  receives one of the $p-1$ hash values in $H$ in all the hash functions
$h_1,\ldots,h_x$ is at most $(\frac{p-1}{b})^{tx}$.
Applying union bound over all such collection of $p-1$ hash values, 
the probability that each vertex of any hyperedge $E$ from $G$  receiving one of the $p-1$ hash values from $\binom{b}{p-1}$ possible $(p-1)$-sized sets in all the hash functions $h_1,\ldots,h_x$ is at most $\binom{b}{p-1} (\frac{p-1}{b})^{tx} < 1$.
It follows that there exists a set of $x$ hash functions from $[n]$ to $[b]$  that constitute a $(G,p,b)$ system.

From $\binom{b}{p-1}(\frac{p-1}{b})^{tx}< 1$, solving for $x$, we get $x > (\frac{p-1}{\ln b- \ln(p-1)}+(p-1))\frac{1}{t}$.
This implies that for any $x=(\frac{p-1}{\ln b- \ln(p-1)}+(p-1))\frac{1}{t}+\epsilon$,
there exists a set of $x$ hash functions from $[n]$ to $[b]$  that constitute a $(G,p,b)$ system.
Choosing the minimum values for $b$ and $t$, we get $x \leq 2$. So, as long as $b > (p-1)(t+1)$,  for every $t$-intersecting hypergraph $G$, there exists a pair of hash functions 
from $[n]$ to $[b]$  that constitute a $(G,p,b)$ system.
\qed 
\end{proof}

The interesting case that remains open is when $b \leq (p-1)(t+1)$ where the above probabilistic method becomes useless.
The bottleneck in analysis $b > (p-1)(t+1)$ comes directly from its dependence on the measure $\mu_p(\mathcal{F})$ given by Theorem 
\ref{thm:2}. Moreover, it is also worthwhile studying the case when $b \leq (p-1)(t+1)$ taking the size of the hypergraph  into consideration.

\bibliographystyle{ieeetr}

\end{document}